\documentclass[11pt]{article}

\usepackage[ruled]{algorithm2e}

\usepackage{fullpage}
\usepackage{amsmath}
\usepackage{amsthm}
\usepackage{amssymb}
\usepackage{hyperref}
\usepackage{graphicx}

\newcommand{\NP}{\mathrm{NP}\xspace}
\newcommand{\PT}{\mathrm{P}\xspace}

\newcommand{\Pp}{\mathcal{P}}
\newcommand{\Qq}{\mathcal{Q}}

\newcommand{\ra}{\rightarrow}

\newcommand{\WL}[1]{\mathrm{WL}(#1)}
\newcommand{\tup}[1]{\mathbf{#1}}
\newcommand{\kequiv}{\equiv^k}

\newcommand{\txor}{\textsc{3-xor}\xspace}
\newcommand{\tsat}{\textsc{3sat}\xspace}

\newcommand{\FF}{{\mathbb F}} 
\newcommand{\defeq}{:=}

\newcommand{\xor}{\oplus}
\usepackage{xspace}
\usepackage{authblk}
\usepackage{pgfplots}
\usepackage{tikz}
\usetikzlibrary{arrows.meta,
	chains,
	positioning,
	shapes.geometric
}
\usepackage{smartdiagram}
\usepackage{subcaption}
\usepackage{booktabs}

\newtheorem{lemma}{Lemma}
\newtheorem{theorem}{Theorem}

\begin{document}

\title{\textbf{Constructing Hard Examples for Graph
    Isomorphism}\thanks{The research reported here was carried out
    while the second author was a student at the University of
    Cambridge during the academic year 2016--17.  The research of the first author is supported by the Alan Turing Institute under the EPSRC grant EP/N510129/1.}}

\author{Anuj Dawar}
\author{Kashif Khan}
\affil{Department of Computer Science
  and Technology\\ University of Cambridge}
\date{}

\maketitle

\begin{abstract}
  We describe a method for generating graphs that provide difficult
  examples for practical Graph Isomorphism testers.  We first give the
  theoretical construction, showing that we can have a family of
  graphs without any non-trivial automorphisms which also have high
  Weisfeiler-Leman dimension.  The construction is based on
  properties of random 3XOR-formulas.  We describe how to convert such
  a formula into a graph which has the desired properties with high
  probability.  We validate the method by experimental
  implementations.  We construct random formulas and validate them
  with a SAT solver to filter through suitable ones, and then convert
  them into graphs.  Experimental results demonstrate that the
  resulting graphs do provide hard examples that match the hardest
  known benchmarks for graph isomorphism.
\end{abstract}

\section{Introduction}
\emph{Graph Isomorphism} (GI) is the problem of deciding, given two graphs
$G$ and $H$ whether there is a bijection between their sets of vertices
$V(G)$ and $V(H)$ respectively, that takes edges of $G$ to edges of
$H$ and non-edges of $G$ to non-edges in $H$.  In short, it asks if
$G$ and $H$ are the same, up to a renaming of vertices.  The problem
is of great interest in the field of complexity theory as it is among
the few natural problems in $\NP$ that are not known to be in $\PT$
nor known to be $\NP$-complete.  Babai's recent result~\cite{Bab16}
places the problem in \emph{quasi-polynomial} time, further cementing
its status as a candidate \emph{$\NP$-intermediate} problem.

While the complexity of GI is interesting from the theoretical
standpoint, in practice the problem is largely solved.  That is, there
exist programs which efficiently deal with instances of graph
isomorphism that arise in practice, for instance in searching through
chemical databases or in image processing.  Significant among these
effective graph isomorphism testers are the programs \texttt{Traces} and
 \texttt{nauty}, available in a common distribution (see~\cite{McKayPiperno}).  It remains a challenge for the
theoretician to examine the algorithms behind these programs and
determine their worst-case behaviour.  Indeed, in the long version of
his paper, Babai asks the question~\cite[Sec.~13.5]{Bab-arxiv}: 
\begin{quote}
   The question is, does there exist an infinite family of pairs of
 graphs on which these heuristic algorithms fail to perform
 efficiently? The search for such pairs might turn up interesting
 families of graphs.
\end{quote}
We address this question and provide a means of constructing just such
a family of graphs, including an implementation and experimental results.

Computationally, the graph isomorphism problem is equivalent to the problem of
determining the orbits of the automorphism group of a graph $G$.  That
is, given a graph $G$, we wish to partition $V(G)$ into the minimum
number of classes so that for any pair of vertices $u,v$ in the same
class there is an automorphism of a graph $G$ that takes $u$ to $v$.
We call this the \emph{orbit partition} of the graph $G$.  The
fundamental algorithm underlying \texttt{nauty} as well as
\texttt{Traces}, like many practical
approaches to the graph isomorphism problem, relies on steadily
refining a partition to arrive at the partition into orbits.  It does
this through a process of \emph{(i)} \emph{vertex refinement} combined
with \emph{(ii)} 
\emph{individualization} and \emph{(iii)} factoring of automorphisms of the graph.
The process of vertex refinement, also known as colour refinement, is
a highly efficient method that is known to achieve the orbit partition
on \emph{almost all} graphs~\cite{BES} but fails miserably on regular
graphs, for example.  Where vertex refinement fails,
the programs use individualization, which is the process of
selecting (i.e.\ individualizing) a particular vertex and placing it
in its own class and then repeating vertex refinement until
evenutually a partition into singleton sets is obtained.  With
backtracking, the structure of the parition into orbits is revealed.  The main
differences between the various graph isomorphism solvers (not only
\texttt{nauty} and \texttt{Traces} but also \texttt{bliss} and
\texttt{conauto}) are precisely in how the vertices are selected.
This process is inherently exponential in the number of selections that need to be
made.  However, the search space is radically cut down if we can
identify non-trivial automorphisms of the graph and factor the graph
suitably, which \texttt{Traces}, in particular, does effectively.  More details
are given in Section~\ref{sec:background} below.

In theory, the vertex refinement algorithm is subsumed by the
$k$-dimensional Weisfeiler-Leman ($k$-WL) method, which works by refining a
partition of the $k$-tuples of the vertices of a graph $G$
(see~\cite[Sec.~2]{CFI92} for a good account of the history of this
method).  Taking $k$ to be large enough (of the order of the number of
vertices in $G$), the $k$-WL method gives exactly the orbit partition,
but this is (for $k \geq 2$) not a practical method and rarely used in
solvers.  However, what is interesting from the point of view of
worst-case analysis is that the $k$-WL method serves as a way of
bounding the number of individualizations we need to determine the
orbit partition in a graph $G$.  To be precise, suppose $G$ has $k$
vertices $v_1,\ldots,v_k$ such that when each of them is
individualized, the vertex refinement procedure converges to the discrete
partition of $G$, then we can also determine the orbit partition by
the $(k+2)$-WL method.  Since it is known, through the construction of
Cai, F\"urer and Immerman~\cite{CFI92} (called CFI graphs below), that there are, for every $k$,
graphs on which the $k$-WL method fails to give the orbit partition,
it follows that there is no constant bound on the number of
individualizations needed, in combination with vertex refinement, to
obtain the orbit partition of a graph.  Hence, there is no polynomial
bound on the running time of a graph isomorphism algorithm based solely on vertex
refinement and individualization.

However, \texttt{Traces} has another element in its armoury, and
that is that it detects automorphisms while constructing the orbit
partition, using these to factor the graph and therefore cut down the
search space.  This means that the running time is not exponential in
the number of individualizations but is potentially divided by the size of the automorphism group of the graph.  Indeed, the CFI graphs are among the standard benchmarks considered
in~\cite{McKayPiperno} and they prove to be not too difficult for the
program to handle as they have many automorphisms.
This led the first author of the
present paper to suggest (at the December 2015 Dagstuhl seminar on
Graph Isomorphism) that the way to construct hard families of graphs
and answer Babai's quetsion,
is to find graphs whose Weisfeiler-Leman dimension is large but which
have no non-trivial automorphisms.  A construction of structures
satisfying this property, known as \emph{multipedes} is given in the
work of Gurevich and Shelah~\cite{GS96}.  These structures can be
easily turned into graphs to yield the desired family.  However, while
the theoretical construction guarantees the existence of such graphs,
it turns out that constructing actual instances, even for small values
of $k$, leads to very large graphs.  Thus, in order to construct
families that can be used for practical benchmarking of GI solvers, a
refined analysis is required.  One such approach was taken by Neuen
and Schweitzer in~\cite{NeuenSchweitzer} where the multipede construction was
combined with \emph{size reduction} techniques.

In the present paper, we give an alternative construction of such
graphs which proves very effective.  Instead of the multipedes of
Gurevich and Shelah, we start with random systems of equations over
the 2-element field.  This is based on the insight from~\cite{ABD09}
that the construction of CFI graphs really codes such systems in the
graph construction.  We use a way of lifting these systems to graphs
which have the property that as long as the original system has no
non-trivial solutions, the resulting graph has no non-trivial
automorphisms.  Moreover, as long as the original system is
$k$-locally satisfiable (a precise definition will follow), the orbit
partition of the graph cannot be obtained by the $l$-WL method for some large $l$.  It
turns out that a random system has both properties: of having no
non-trivial solutions and being $k$-locally satisfiable for sublinear
values of $k$.  The theoretical basis of this construction is given in
Section~\ref{sec:construction}.  The main conclusion is
Theorem~\ref{thm:asymmetric} which shows that we can construct
families of graphs which are \emph{asymmetric}, i.e.\ have no
non-trivial automorphisms, but have \emph{linear} Weisfeiler-Leman
dimension.  Moreover, we define a randomized construction that
produces such graphs with high probability.

It is instructive to compare our theoretical result with that of Neuen
and Schweitzer~\cite{NeuenS-ESA,NeuenSchweitzer}.  They also give a
construction of families of graphs which provably require exponential
time on a solver based only on individualization and refinement.
Here, they take refinement to be any class of procedures that respect
$k$-WL equivalences for some fixed $k$.  Their construction is based
on the multipede construction of Gurevich and Shelah, converted into
graphs.  This construction guarantees that the graphs are asymmetric,
and have unbounded Weisfeiler-Leman dimension.  While the dimension is
unbounded, the multipede construction does not yield \emph{linear}
dimension.  Indeed, as noted above, the graphs obtained for small
values of $k$ are rather large.   Therefore, in order to obtain exponential lower bounds
Neuen and Schweitzer employ size reduction techniques and explicitly consider the
shape of the individualization search tree.  By contrast, we
directly establish a linear lower bound on the WL-dimension of the
graphs and this immediately leads to an exponential lower bound on the
search tree size under any target cell selection strategy.  The
experimental results reported in~\cite{NeuenS-ESA} are comparable with
the ones we obtain.

To implement our theoretical construction, we leverage a highly
developed SAT solver.  This enables us to search for systems of
equations which have no non-trivial solutions by coding them as 3SAT
instances.  While we do not directly check for $k$-local
satisfiability, we use a proxy which is checking the speed improvement
in the SAT solver that is obtained by the use of Gaussian elimination
methods.  This filter allows us to select those systems which are most
likely to be locally satisfiable.  We present details of the method in
Section~\ref{sec:setup}.  Finally, we present some experimental
results in Section~\ref{sec:results} which show that the method does,
indeed, yield instances which are hard, especially for
\texttt{Traces}, but also for other isomorphism solvers.  We
have created some benchmark sets of such graphs, and one of these is
now available through the webpage for \texttt{nauty} and
\texttt{Traces}: \url{http://pallini.di.uniroma1.it/Graphs.html}.  However, we consider the
main contribution of the present work to be the method, using a SAT
solver, which gives the ability to generate such large and hard
example graphs at will. 

The experimental work reported here was carried out in three
stages.  The first set of experiments were performed in April-June
2017 as part of the
second author's Master's project.  The full code of the
implementation, all data generated in the experiments, as well as a
project write-up can be found here:
\url{https://github.com/kkcam/graph-ismorphism}.  In particular, a
number of graphs in \texttt{dreadnaut} format (\texttt{.dre}) can be
downloaded from the site to run directly with \texttt{nauty/Traces}.
A second set of experiments on the same graphs, involving a wider range of solvers, was
carried out in January-February 2019, in response to a request from a
referee.  The results are available from the same website.  Finally,
another set of graphs using this protocol was created by Yui Chi
Richie Yeung.  Those graphs and the results are available at
\url{https://github.com/y-richie-y/sat_cfi/}.  A selection of results from all
three sets of experiments is given in Section~\ref{sec:results}.

The theoretical work
detailed in Section~3 is based on the first author's project
suggestion, given as an appendix in the Master's thesis available at
\url{https://github.com/kkcam/graph-ismorphism}.

\paragraph{Acknowledgements}  We would like to thank Richie Yeung for
allowing us to report some of his experimental results in this paper.
We would also like to thank Jakub Rydval for helping to correct an
error in an earlier version of Lemma~\ref{lem:unique}.

\section{Background}\label{sec:background}

\paragraph{Partitions and Isomorphisms.}
Given a set $S$ and two partitions $\Pp = \{P_1,\ldots,P_s\}$ and $\Qq
= \{Q_1,\ldots,Q_t\}$ of $S$, we say that $\Pp$ is a \emph{refinement} of
$\Qq$ (or equivalently, that $\Qq$ is \emph{coarser} than $\Pp$) if
for every $P \in \Pp$ there is a $Q \in \Qq$ such that $P \subseteq
Q$.  It is a \emph{proper refinement} if $s > t$.  We say that a
partition $\Pp$ is \emph{discrete} if every part is a singleton.

We always consider undirected, loopless, simple graphs.  That is, a
graph $G$ is a set of vertices $V(G)$ along with a set of edges $E(G)$
where each edge $e \in E(G)$ is a two-element set $e = \{u,v\}
\subseteq V(G)$, with $u \neq v$.  Where $G$ is clear from context, we
simply write $V$ and $E$ for the vertex and edge set respectively.
For a set $C$, a $C$-coloured graph is a a graph $G$ together with a
function $\chi: V \ra C$.  For a $C$-coloured graph $(G,\chi)$, we refer to the
partition of $V$ given by $\{ \{v \mid \chi(v) = c \} \mid c \in C\}$ as
the $\chi$-partition of $V$.

Given two $C$-coloured graphs $(G,\chi)$ and $(H,\delta)$, an
\emph{isomorphism} from the first to the second is a bijection $\iota : V(G)
\ra V(H)$ such that for each $u,v \in E(G)$, $\{u,v\} \in E(G)$ if,
and only if, $\{\iota(u), \iota(v)\} \in E(H)$ and $\chi(v) =
\delta(\iota(v))$.  An \emph{automorphism} of $(G,\chi)$ is an
isomorphism from   $(G,\chi)$ to itself.  The \emph{orbit partition}
of $(G,\chi)$ is the coarsest partition of $V(G)$ such that if $u$ and $v$
are in the same part of the partition, there is an automorphism
$\iota$ of $(G,\chi)$ such that $\iota(u) = v$.  Consider the three
computational problems: (1) given a pair of graphs $G$ and $H$, decide
if there is an isomorphism from $G$ to $H$; (2) given a pair of
coloured graphs $(G,\chi)$ and $(H,\delta)$ decide if there is an
isomorphism between them; and (3) given a coloured graph $(G,\chi)$,
output its orbit partition.  It is known that these three problems are
polynomially-equivalent, which is to say that there are
polynomial-time reductions between any pair of them (see, for
instance~\cite{Toran04}).  As such, we treat them as equivalent and
mostly concentrate on the third.

A graph $G$, or a coloured graph $(G,\chi)$, is called
\emph{asymmetric} if its only automorphism is the identity map.  Some
authors call such graphs \emph{rigid}, but we employ the terminology
of Babai~\cite{Babai-handbook} who reserves the latter term for graphs
without non-trivial endomorphisms.

\paragraph{Refinement and Individualization.}
The \emph{vertex refinement} procedure is an algorithm that produces,
given a $C$-coloured graph $(G,\chi)$, the coarsest partition $\Pp$
of $V$ refining the $\chi$-partition such that if $u,v \in P \in \Pp$, then for
every $Q \in \Pp$, $u$ and $v$ have the same number of neighbours in
$Q$.  Note that the partition of $V$ produced by the vertex refinement
procedure is coarser than the orbit partition.  For many graphs $G$
it is, in fact, the orbit partition but, for example for regular
graphs with no colouring, it can be properly coarser
(see~\cite{BES}). 

Given a $C$-coloured graph $(G,\chi)$ and a vertex $v \in V$, let $c$
be a new colour that does not appear in $C$ and let $\chi':V
\rightarrow C \cup \{c\}$ be defined
as the colouring with  $\chi'(v)  = c$ and $\chi'(u) = \chi(u)$ for $u
\neq v$.  Then, the partition of $V$ obtained by vertex refinement
from $(G,\chi')$ is called the vertex refinement of  $(G,\chi)$ with
\emph{individualization} of $v$.  More generally, given a sequence
$I = (v_1,\ldots,v_k)$ of vertices in $V$, the vertex refinement with
individualizations of $I$ is the algorithm that produces the vertex
refinement of $(G,\chi')$ where $\chi'$ is a
$C \cup\{c_1,\ldots,c_k\}$-colouring of $V(G)$ with
$\chi'(v_i) = c_i$ and $\chi'(v) = \chi(v)$ for $v \not\in I$; and
$c_1,\ldots,c_k \not\in C$.  This
partition is \emph{not} in general coarser than the orbit partition.
The aim of the individualization-refinement procedure is to find the
smallest (in a precise sense)  $I$ such that the refinement with
individualizations of $I$ yields a partition into singleton sets.  From this, it
is possible to produce the orbit partition of $(G,\chi)$.  For
details, we refer the reader to~\cite[Sec.~2]{McKayPiperno}.

\paragraph{High-dimensional Weisfeiler-Leman.}
For each integer $k \geq 2$, the $k$-dimensional Weisfeiler-Leman
algorithm gives a partition of $V^k$ that is the coarsest partition
$\Pp$ satisfying the following stability condition: if $\tup{u},\tup{v} \in V^k$
are tuples in the same part of $\Pp$ and $(P_1,\ldots,P_k)$ is a
$k$-tuple of parts of $\Pp$, then the order-preserving map from
$\tup{u}$ to $\tup{v}$ is an isomorphism of the induced subgraphs of
$G$ and $|\{ x \mid \tup{u}[x/i] \in P_i \text{ for }1\leq  i \leq k\}| = |\{
x \mid \tup{v}[x/i] \in P_i  \text{ for }1\leq  i \leq k\}|$.  Here, $\tup{u}[x/i]$ denotes the
tuple obtained by substituting $x$ for the $i$th element of $\tup{u}$
and $|S|$ denotes the cardinality of a set $S$.  We write $\kequiv$ to
denote the equivalence relation corresponding to this partition.
Again, this partition is necessarily coarser than the partition of
$V^k$ into orbits under the action of the automorphism group of $G$,
since the orbit partition clearly satisfies the stability condition.
Also, for sufficiently large values of $k$, in particular certainly
for $k \geq n-1$, the partition given by $\kequiv$ \emph{is} the orbit
partition.  So, for any graph $G$, we define the
\emph{Weisfeiler-Leman dimension} of $G$, denoted $\WL{G}$, to be the
least $k$ such that the partition induced by $\kequiv$ on $V^k$
coincides with the orbit partition.

Cai, F\"urer and Immerman~\cite{CFI92} showed that there is no fixed $k$
such that $\WL{G} < k$ for all graphs $G$.  Indeed, they show a linear
lower bound on $\WL{G}$.  To be precise, they construct for each $k$ a pair $G$ and $H$ of
non-isomorphic graphs with $O(k)$ vertices such that $G$ and $H$
cannot be distinguished by the isomorphism test based on
$k$-dimensional Weisfeiler-Leman equivalence.  This implies that
$\WL{G\uplus H} > k$, where $G\uplus H$ is the disjoint union of $G$
and $H$.  Through this operation of disjoint union, the problem of
testing graph isomorphism reduces to that of constructing the orbits
of the automorphism group and, as stated in the introduction, here we
adopt the latter perspective.

The $k$-dimensional Weisfeiler-Leman isomorphism test has been
extensively analyzed in theoretical studies of the graph isomorphism
problem.  It has many equivalent and strikingly different
characterizations, arising in algebra, combinatorics, logic and
optimization.  In particular, it is known that in a graph $G$,
$\tup{u} \kequiv \tup{v}$ if, and only if, there is no formula of
$C^{k+1}$ (first-order logic with counting quantifiers using at most
$k+1$ distinct variables) that distinguishes $\tup{u}$ from $\tup{v}$.
This relation was given a useful characterization in terms of a $k$-pebble
\emph{bijective game} by Hella~\cite{Hel92}.  The game is played by
two players called Spoiler and Duplicator on a graph $G$.  The
position of the game at any moment in time consists of two $k$-tuples
of vertices $\tup{u}$ and $\tup{v}$.  In each move, Spoiler chooses a
value of $i\in [k]$ and Duplicator responds with a bijection
$\pi: V \ra V$.  Spoiler then chooses a vertex $x \in V$ and the new
position is $\tup{u}[x/i]$ and $\tup{v}[\pi(x)/i]$.  Spoiler wins in
any position if the \emph{ordered} subgraph of $G$ induced by the
tuple $\tup{u}$ is not isomorphic to the \emph{ordered} subgraph
induced by $\tup{v}$.  The result of Hella~\cite{Hel92} essentially
tells us that Duplicator has a strategy to play forever without
Spoiler winning in the $(k+1)$-pebble bijective game starting at
position $\tup{u},\tup{v}$ if, and only if, $\tup{u} \kequiv \tup{v}$.

The connection between the Weisfeiler-Leman dimension of a graph and
the refinement and individualization procedure is the following.  If a
graph $G$ contains $k$ vertices $v_1,\ldots,v_k$ such that
individualizing them results in the vertex refinement procedure on $G$
producing the discrete partition, then $\WL{G} \leq k+2$.  This is
most easily seen through the characterization of the Weisfeiler-Leman
dimension in terms of counting logic.  That is, when such a set
of $k$ vertices exist, we can show that any pair $x$ and $y$ of
vertices that are not distinguished by any formula of $C^{k+3}$ are in
the same orbit of the automorphism group of $G$.  The argument is as
follows.  We know that the vertex refinement procedure yields a
partition into $C^2$-equivalence classes.  By the assumption that individualizing $v_1,\ldots,v_k$
results in the discrete partition, we have that for each vertex $x$,
there is a formula $\phi(x)$ of $C^2$ using constants for
$v_1,\ldots,v_k$ that is only true of $x$ and of no other vertex in
$G$.  Write $\phi_x$ for the formula of $C^{k+2}$ where the constants
in $\phi$ have been replaced with new variables (which we will also
denote $v_1,\ldots,v_k$ for convenience).  Now consider the formula 
$$ \theta( v_1,\ldots,v_k) \defeq \bigwedge_{x \in V(G)}
\exists^{=1}x \phi_x \land \forall x,y \big( E(x,y) \leftrightarrow
\bigvee_{\{x,y\} \in E(G)} (\phi_x \land \phi_y)\big) .$$
Note, here, while we have a different formula $\phi_x$ for each $x \in
V(G)$, we assume that we re-use the variable $x$.  However, where we
write $\phi_y$, we replace it with the new variable $y$.  This ensures
that the total number of variables used is at most $k+3$.  It can then
be verified that the formula $\exists v_1,\ldots,v_k \theta$ describes
the graph $G$ uniquely, up to isomorphism.  Moreover, $\exists
v_1,\ldots,v_k (\theta \land \phi_x)$ is a formula of $C^{k+3}$ that
is only satisfied in $G$ by vertices in the orbit of $x$.

\paragraph{Multipedes.}
Gurevich and Shelah~\cite{GS96} show how to construct a class of
finite structures which is (1) axiomatized by a sentence of
first-order logic; (2) contains only structures with no non-trivial
automorphisms; and (3) such that no formula of fixed-point logic defines a
linear order on all structures in this class.  From our point of view,
the relevant aspect of this construction is that it gives, for each
value of $k$, an \emph{asymmetric} structure $M_k$, that is one which has no
non-trivial automorphisms, but on which the partition into
$\kequiv$-classes is non-trivial.  That is to say, even though every
element of the orbit partition of $M_k$ is a singleton, there are
pairs of distinct elements $a,b$ in $M_k$ such that $a \kequiv b$.
The structures in question are called $3$-multipedes in~\cite{GS96}.

It seems at first sight that this provides suitable hard examples for
graph isomorphism testers such as \texttt{Traces}.  Strictly speaking, the
$3$-multipedes are not graphs, but they can be translated to graphs by
standard methods (see~\cite[Theorem~5.5.1]{Hodges}), preserving the relevant
properties: unbounded Weisfeiler-Leman dimension and no non-trivial
automorphisms.  The absence of non-trivial automorphisms means that
factoring by automorphisms cannot be used to speed up search by
trimming the tree, while the unbounded dimension guarantees that there
is no upper bound on the number of individualization steps needed to
make the vertex refinement procedure yield the discrete partition.
However, actually constructing instances of such multipedes turns out
to be difficult.  The proof in~\cite{GS96} does not actually show how
to construct the structures $M_k$.  Rather, it shows that for all
large enough values of $n$, a random structure on $n$ elements, under
a suitable skew probability distribution $\mu_n$, has the right properties.
However, the probability grows rather slowly with $n$.  Indeed, the
smallest value of $n$ at which the probability is non-zero is possibly
exponential in $k$.  An experimental attempt to sample from the distribution
$\mu_n$ failed to produce interesting examples at values of $n$ up to
a few thousand~\cite{Arrighi-git}.

In the present paper, we consider an alternative construction, based
on similar principles, of graphs whose Weisfeiler-Leman dimension is
linear in the number of vertices (as with the CFI
graphs), and which have no non-trivial automorphisms.  The
construction is again randomized, but based on a simple and
well-understood probability distribution.  Furthermore, a use of a SAT
solver enables the quick generation of examples by filtering graphs
that are sampled from the distribution.

\section{The Construction}\label{sec:construction}
In this section, we describe a construction that yields, for each $k$,
a graph $G_k$ with $O(k)$ vertices with the property that $G_k$ is
asymmetric and has Weisfeiler-Leman dimension at least $k$.   The
proof that such graphs exist is derived from known results in the literature, and here
we show how to derive it, giving the necessary definitions to
understand the background.  The starting point of the construction is
the observation that we can define instances of \txor that
are $k$-locally consistent but unsatisfiable.

\paragraph{XOR formulas.}
Fix a countable set $\mathcal{X}$ of Boolean variables.  We use
capital letters $X, Y, \ldots$ to range over this set.  A
$\txor$-formula is a finite set of clauses, where each clause contains
exactly $3$ literals, each of which is either a variable $X$ or a
negated variable $\bar{X}$.

We say that a $\txor$-formula $\phi$ is satisfiable if there is an
assignment $T: \mathcal{X} \ra \{0,1\}$ of truth values to the
variables $\mathcal{X}$ such that in each clause of $\phi$, an
\emph{even} number of literals is made true.

Given a $\txor$-formula $\phi$, we can construct a system of linear
equations over the two-element field $\FF_2$.  That is, for each
clause $C$ of $\phi$ we construct the equation $x + y + z = c$ where
$x,y,z$ are the variables occurring in the literals of $C$ and $c$ is
$1$ if an odd number of them appear negated and $0$ otherwise.  It is
easily verified that this system of equations has a solution if, and
only if, $\phi$ is satisfiable.  Note that two distinct clauses may
give rise to the same equation.  Say that two clauses are
\emph{equivalent} if they give rise to the same equation.  In the
sequel, we will use the terminology of $\txor$ formulas and of systems
of linear equations interchangeably, according to which is convenient.

So, we can think of a $\txor$ formula with $n$ variables and $m$ clauses as a system of equations $H \tup{x} =
\tup{b}$ where $H$ is an $m \times n$ matrix, $\tup{x}$ the tuple of
$n$ variables and $\tup{b} \in \FF_2^m$ the $m$-tuple of right-hand
sides of the equation.  We say the system is \emph{homogeneous} if the
right-hand side of every equation is $0$.  This corresponds to a
$\txor$ formula in which no variable appears negated.  A homogeneous system is
always satisfiable, as it is satisfied by the assignment of $0$ to
every variable.  Note that a homogeneous system is completely
specified by a collection of $3$-element sets of variables, with one
set for each equation, containing the three variables that appear in
it. 

\paragraph{Random XOR formula.} 
For fixed positive integers $m,n$ we write $F(m,n)$ for the set of all
$\txor$-formulas over the variables $X_1,\ldots,X_n$ containing
exactly $m$ inequivalent clauses.  We also write $\mathcal{F}(m,n)$
for the uniform probability distribution over $F(m,n)$.  It is known
that, for large enough values of $m$ and $n$, with $m > n$, a random
formula drawn from this distribution is unsatisfiable
(see~\cite{PittelS16}).   That is to say that as $n$ increases, for
all $m > n$, the probability that a formula drawn from the
distribution $\mathcal{F}(m,n)$ is satisfiable tends to $0$.

We are interested in $\txor$-formulas that are unsatisfiable but
$k$-locally consistent, for suitable integer $k$.  For our purposes,
we define $k$-local consistency by means of the following pebble game,
played by two players called Spoiler and Verifier.  The game is played
on a $\txor$-formula $\phi$ with $k$ pebbles $p_1,\ldots,p_k$.  At
each move Spoiler chooses a pebble $p_i$ (either one that is already
in play, or a fresh one) and places it on a variable
$X$ appearing in $\phi$.  In response, Verfier has to choose a
value from $\{0,1\}$ for the variable $X$.  If, as a result, there is
a clause $C$ such that all literals in $C$ have pebbles on them and
the assignment of values to them given by Verifier results in $C$
being unsatisfied, then Spoiler has won the game.  Otherwise the game
can continue.  If Verifier has a strategy to play the game forever
without losing, we say that $\phi$ is $k$-locally consistent.

It is also known that for all $k$, the probability that a random
formula drawn from $\mathcal{F}(m,n)$ (with $m> n$) is $k$-locally consistent tends
to $1$ as $n$ increases.  This was proved for formulas of \tsat rather
than \txor in~\cite{Ats04}, but a similar analysis shows the result
also for $\txor$.  Such an analysis can be found in~\cite[Lemma~4]{AD18}. 

\paragraph{Unique Satisfiability.}
As noted above, a homogeneous system of equations is
always satisfiable, as it is satisfied by the assignment of $0$ to
every variable.  We say that the system is \emph{uniquely satisfiable}
if this is the only solution to the system.   It is easy to see that
the set of solutions to $H \tup{x} = 0$ is exactly the null
space of the matrix $H$, as a subset of the vector space $\FF_2^n$.  In particular, the system is uniquely
satisfiable if, and only if, $H$ has rank $n$.

Define $H(m,n)$ to be the set of all homogeneous systems of equations
with $m$ clauses and $n$ variables, and $\mathcal{H}(m,n)$ for the
uniform probability distribution over this set.   We use the following fact about this distribution, established in~\cite{BKK92}
\begin{lemma}\label{lem:unique}
  There is a threshold $t > 1$ such that, for any $\alpha > t$, the probability that a random system
  drawn from $\mathcal{H}(\alpha n,n)$ is uniquely satisfiable tends
  to $1$ as $n$ increases.
\end{lemma}

A homogeneous system is necessarily satisfiable, and so also
$k$-locally consistent for all $k$.   If it is uniquely satisfiable, it has no solutions when we require some fixed variable $X_i$ to take value $1$.  However, for a randomly chosen such system, adding the condition $X_i = 1$ leaves it 
 $k$-locally consistent for small values of $k$.  
To be precise, there is a
constant $\gamma$ such that if $\phi$ is a system $H  \tup{x} = 0$ chosen at random from
$\mathcal{H}(\alpha n,n)$, then with high probability there is some
$i$ such that if $\phi_i$ is the system obtained from $\phi$ by adding the equation $X_i = 1$ 
then $\phi_i$ is $\gamma n$-locally
consistent.  This is a much weaker statement than proved in~\cite[Lemma
4]{AD18} where it is shown that changing the right-hand sides of a
random subset of the vertices to $1$ still leaves the system
$\gamma n$-locally consistent.

\paragraph{CFI construction}
The construction of Cai, F\"urer and Immerman~\cite{CFI92} provides us
with examples of pairs of non-isomorphic graphs which are not
distinguished by the $k$-dimensional Weisfeiler-Leman test.
Inspired by this, a construction described in~\cite[Prop.~32]{ADW17} shows how
to convert any $k$-locally consistent system of equations in
$H(m,n)$ to one that cannot be distinguished by the $k$-dimensional
Weisfeiler-Leman test from its homogeneous companion.  Here, the
\emph{homogeneous companion} of a system $H \tup{x} =\tup{b}$ is $H
\tup{x} =\tup{0}$ (see also~\cite[Lemma 3]{AD18} for a similar
argument).  Here we adapt the construction, to turn an arbitrary 
system $\phi$ into a graph $G_\phi$ with the property that the local
consistency of
$\phi$ translates into a lower bound on the Weisfeiler-Leman dimension
of $G_\phi$.  Moreover, the \emph{unique satisfiability} of $\phi$
guarantees that $G_\phi$ is asymmetric.

For any $\txor$-formula $\phi$, we define the graph $G_{\phi}$ by the
following construction.  If $\phi$ has $m$ inequivalent clauses and
$n$ variables, $G_{\phi}$ has a total of $4m + 2n + 3(n-1)$ vertices.

Let $X_1,\ldots,X_n$ be the $n$ variables in some fixed order.  For
each clause $C$ of $\phi$, we define the four clauses $C_{000},
C_{011}, C_{110}$, and $C_{101}$ by letting $C_{000} = C$ and
$C_{011}, C_{110}, C_{101}$ be
the three clauses equivalent to $C$ obtained by negating exactly two
of the literals of $C$.  In particular, $C_{011}$ is obtained from
$C_{000}$ by negating the second and third variables in the clause,
$C_{110}$ by negating the first and second and $C_{101}$ by negating
the first and third.  Here, the terms ``first'', ``second'' and
``third'' refer to the numberical order of the variables chosen above.

We then have a vertex in $G_{\phi}$ for each
of these clauses.  Also, for each variable $X$ in $\phi$, we have two
vertices $X^0$ and $X^1$.  In addition, for each $i$ with $1 \leq i <
n$ we have three vertices $i_l,i_r,i_s$.

\begin{figure}[h]
\centering
  \def\svgwidth{0.5\textwidth}
  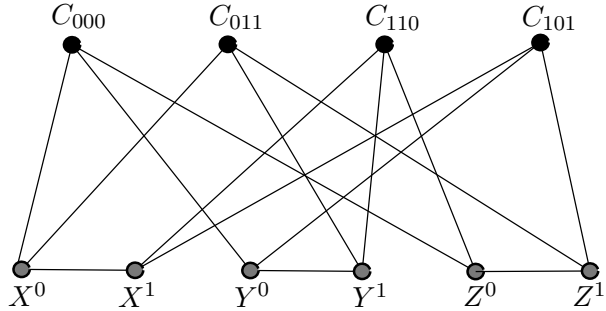
\caption{Clause gadget in $G_{\phi}$ corresponding to the clause $X
  \xor Y\xor Z$}\label{fig:cfi-gadget}
\end{figure}
The edges are described as follows.  For each clause $C$, if the
literal $X$ occurs in $C$, we have an edge from $C$ to $X^1$ and if
the literal $\bar{X}$ occurs in $C$, we have an edge from $C$ to
$X^0$.  There is an edge between $X^0$ and $X^1$.  These are depicted
in Figure~\ref{fig:cfi-gadget}.  These capture the essence of the CFI-like
construction.  In addition, for each $i$ we also have the edges: $(i_l,i_r)$, $(i_r,i_s)$
and $(i_l,X_i^0)$, $(i_l,X_i^1)$, $(i_r,X_{i+1}^0)$ and
$(i_r,X_{i+1}^1)$.  These additional edges are depicted in
Figure~\ref{fig:rigid}.  The purpose of the additional gadget
involving the vertices $i_l,i_r$ and $i_s$ is to remove some
symmetries on the graph by imposing the chosen order on the set of
variables. 
\begin{figure}[h]
\centering
  \def\svgwidth{0.4\textwidth}
  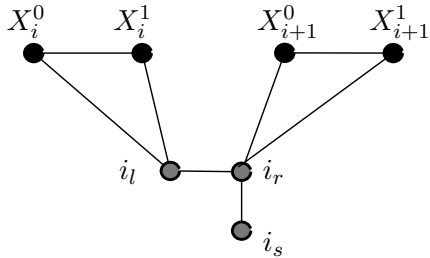
\caption{Asymmetry gadget in $G_{\phi}$} \label{fig:rigid}
\end{figure}

Now, fix a homogeneous system of equations $\phi$, and
let $G_\phi$ be the graph obtained by the above construction.  Also,
let $\phi_i$ be the system obtained by adding the equation $X_i = 1$ to $\phi$.

\begin{lemma}\label{lem:consistency}
  If $\phi_i$ is $3(k+1)$-locally consistent, then $X_i^0 \equiv^k X_i^1$ in $G_{\phi}$.
\end{lemma}
\begin{proof} 
The proof follows along the lines of~\cite[Lemma 3]{AD18} by showing
a winning strategy for Duplicator in the bijective $(k+1)$-pebble game played
on $G_{\phi}$ starting in the position $\tup{u}, \tup{v}$ where
$\tup{u}$ is the tuple consisting of the vertex  $X_i^0$ repeated $k+1$
times and $\tup{v}$ is the tuple consisting of the vertex  $X_i^1$ repeated $k+1$
times.  We give a brief outline.

Duplicator's strategy will always be to play a bijection that is the
identity on the vertices  $i_l,i_r,i_s$.  For each variable $X$ it
maps the set $\{X^0, X^1\}$ to itself (though it may swap these two
vertices) and for each clause $C$ it maps the set
$\{C_{000},C_{011},C_{110},C_{101}\}$ to itself (though it may permute these
elements).  Moreover the permutation induced on $\{C_{000},C_{011},C_{110},C_{101}\}$
must be either the identity or a permutation induced by swapping $X^0$
and $X^1$ for exactly two variables $X$ appearing in the clause $C$.
Call a bijection meeting these requirements \emph{well-defined}.

Given a position $\tup{u},\tup{v}$ in the bijective game, we say that
it is \emph{consistent} if there is a well-defined bijection $\beta$
taking $\tup{u}$ to $\tup{v}$ and such that for any $C_i \in \tup{u}$
if $\beta(C_i) = C_j$ where $C_j$ is obtained from $C_i$ by swapping
$X^0,X^1$ and $Y^0,Y^1$ then $\beta(X^0) = X^1$ if either of $X^0$ or
$X^1$ is in $\tup{u}$ and similarly $\beta(Y^0) = Y^1$ if either of $Y^0$ or
$Y^1$ is in $\tup{u}$.

Consider now a consistent position  $\tup{u},\tup{v}$ and let $\beta$
be a well-defined bijection witnessing this.  Let
$U$ be the set of variables of $\phi$ containing all variables $X$ such
that either $X^0$ or $X^1$ appears in $\tup{u}$ or $X$ appears in some
clause $C$ such that one of $C_{000},C_{011},C_{110},C_{101}$ appears in $\tup{u}$.
Note that $U$ has at most $3(k+1)$ elements.  
Now, we define the map $T: U \ra \{0,1\}$ given by $T(X) = 0$ if
$\beta(X^0)=X^0$ and $T(X) = 1$ if $\beta(X^0)=X^1$.  Duplicator's strategy is to ensure that $T$ is a
winning position for Verifier in the $3(k+1)$-local consistency game on the formula
$\phi$.  That is to say, starting in the position where pebbles are
placed on the variables in $U$, and Verifier's responses are given by
$T$, Verifier can continue and play forever.

This condition is satisfied by the initial position, as $\tup{u}$ is
just the element $X_i^0$ repeated $k+1$ times, so $U = \{X_i\}$, and $T$
is the map taking $X_i$ to $1$.  But, the fact that $\phi_i$ is
$k+1$-locally consistent implies that this is a winning position for
Verifier.   Now, to see that the Duplicator can maintain the position,
suppose at some stage in the bijective game, Spoiler moves pebble
$j$.  Duplicator needs to construct a well-defined bijection such that
anywhere Spoiler places the pebble will result in a consistent
position.  Spoiler's move can be translated into a move in the local
consistency game from the current position $T$.  Duplicator's possible
responses in that game define a function from the variables $X$
to $\{0,1\}$ and this can be turned into a well-defined bijection.  
\end{proof}

\paragraph{Asymmetry}
Finally, we want to argue that if the homogeneous system $\phi$ is
uniquely satisfiable, then $G_{\phi}$ is asymmetric.  Before giving
the proof, we give some intuition.  The graph $G_\phi$ has two
vertices $X^0$ and $X^1$ for each variable $X$ of $\phi$.  Consider
first a permutation $\pi$ of these vertices which fixes each set
$\{X^0,X^1\}$.  This gives rise to a map $T$ from the variables of
$\phi$ to $\{0,1\}$ such that $T(X) = 1$ if, and only if, $\pi$
exchanges $X^0$ and $X^1$.  To extend $\pi$ to an automorphism of
$G_\phi$ would require us to permute the vertices corresponding to
clauses in such way that fixed each set $\{C_{000},C_{011},C_{110},C_{101}\}$.  This can only happen if for exactly two of the variables
$X$ appearing in the clause $C$ do we have $T(X) = 1$.  In other
words, this requires $T$ to be an assignment satisfying $\phi$.  Are
there other automorphisms of $G_\phi$ that do not fix the sets
$\{X^0,X^1\}$?  The presence of the vertices $i_l,i_r,i_s$ ensures
that these sets are effectively linearly ordered and no other
automorphisms are possible.  This is formally proved below.

\begin{lemma}\label{lem:asymmetry}
If $\phi$ is homogeneous, then it is uniquely satisfiable if, and only if, 
$G_{\phi}$ is asymmetric.
\end{lemma}
\begin{proof}
Let $\alpha$ be any automorphism of $G_{\phi}$.   Note that every clause vertex $C$ has degree 3.  Every variable vertex $X^0$ or $X^1$ has degree at least 4.  Every vertex $i_s$ has degree 1.  Thus, each of the following sets is fixed (set-wise) by $\alpha$:
\begin{itemize}
\item the set $S= \{i_s \mid 1\leq i < n\}$: this is the set of vertices of degree 1;
\item the set $R = \{i_r \mid 1\leq i < n\}$: this is the set of vertices adjacent to a vertex in $S$;
\item the set of clause vertices $\mathcal{C}$ : this is the set of vertices of degree 3 that are not within distance 2 of a vertex in $S$;
\item the set of variable vertices $\mathcal{X}$: this is the set of neighbours of $\mathcal{C}$; and 
\item the set $L = \{i_l \mid 1\leq i < n\}$: this is everything else.
\end{itemize}

Indeed, we can say more.  Each of the sets $S$, $L$ and $R$ is fixed \emph{pointwise} by $\alpha$.  If this were not so, there would be some $i,j$ with $i < j$ such that $\alpha(i_s) = j_s$ (since the set $S$ is fixed).  Then, $\alpha(i_r) = j_r$ (since these are the sole neighbours), $\alpha(\{X_i^0,X_i^1\}) = \{X_j^0,X_j^1\}$ (since these are the only neighbours in $\mathcal{X}$ of $i_r$ and $j_r$ respectively), and so $\alpha((i+1)_l) = (j+1)_l$ and $\alpha((i+1)_r) = (j+1)_r$.  Proceeding by induction, we have for all $k$ $\alpha((i+k)_r) = (j+k)_r$.  Taking $k$ large enough so that $j+k > n \geq i+k$, we get a contradiction.

It also follows that, for each variable $X$, $\alpha(\{X^0,X^1\}) = \{X^0,X^1\}$.  That is, $\alpha$ either fixes each of the two vertices or it interchanges them.  Note that if $\alpha$ fixes all the variable vertices, then it is the identity everywhere, since no two vertices in $\mathcal{C}$ have the same neighbours in $\mathcal{X}$.  Let $T$ be the assignment that maps $X$ to $0$ if $\alpha$ is the identity on $\{X^0,X^1\}$ and $1$ otherwise.  We can check that $T$ satisfies $\phi$.

In the other direction, suppose there is a truth assignment $T$ that
satisfies $\phi$.  Now consider the map on $\mathcal{X}$ that exchanges
the vertices $X^0$ and $X^1$ just in case $T(X) = 1$ and is the
identity everywhere else.  We extend this to a map  on $\mathcal{C}$ as
follows.  For any clause $C$ of $\phi$, there are either $0$ or $2$ variables
$X$ in $C$ for which $T(X)=1$, since $T$ is a satisfying assignment.
In the first case, we let our map be the identity on
$\{C_{000},C_{011},C_{110},C_{101}\}$ and in the latter case it is the
unique permutation of this set induced by exchanging $X^0$ and $X^1$
for the two variables such that $T(X)=1$.  Finally, we also define the
map  to be the identity on the set $L \cup R \cup S$.  It is now easy
to see that this map is an automorphism.
\end{proof}

We can conclude the description of the construction with the statement
of a theorem.

\begin{theorem}\label{thm:asymmetric}
  There is a family of asymmetric graphs $G_k$ with $O(k)$ vertices
  and Weisfeiler-Leman dimension $k$.
\end{theorem}
\begin{proof}
  By Lemma~\ref{lem:unique}, there is an $\alpha$ such that for sufficiently large $n$, a randomly chosen homogeneous
  system of equations $\phi$ from $\mathcal{H}(\alpha n,n)$ is
  uniquely satisfiable with positive probability.  By Lemma~\ref{lem:asymmetry}, $G_{\phi}$ is
  then asymmetric.  On the other hand, by~\cite[Lemma~5]{AD18},
  $\phi_i$ is $\gamma n$-locally consistent for any $i$ with high probability, which implies
  by Lemma~\ref{lem:consistency} that $G_\phi$ has WL-dimension at
  least $\frac{1}{3}\gamma n - 1$.
\end{proof}
Note that we really have proved, not only the existence of such a
family, we have described a random process that produces such graphs with high probability.

\paragraph{Size Bounds}

While Theorem~\ref{thm:asymmetric} tells us that the graphs $G_k$ have
size linear in $k$, the actual size bounds are somewhat less clear.
Specifically, there is a probabilistic element to the construction
that relies on constructing a uniquely satisfiable formula $\phi$
such that for some $i$, $\phi_i$ is $k$-locally consistent.  What we
know is that for any $k$, and any large enough $n$, a randomly
constructed formula with $n$ variables and $m=\alpha n$ clauses ($\alpha > 1$) will
have these properties with high probability.  How big does $n$ have to
be before the probability becomes significant?  A direct calculation
does not give much cause for optimism.

Our argument for why a random formula is $k$-locally consistent with
high probability is based on~\cite[Lemma 3]{AD18}, which in turn
relies on the argument for expansion of a random $3$-uniform
hypergraph given in~\cite{BSW01}.  The key combinatorial bound in  on width is \cite{BSW01}[Lemma~6.6], attributed to~\cite{Beame}.  A simple
calculation shows that we need $m$ to be around $10^7$ in order to be
guaranteed a width lower bound of $2$ (i.e.\ that a formula will be
$2$-locally consistent with probability greater than $1/2$).  With $m$
around $10^9$, we get reasonably high bounds on width, but these would
be much larger graphs than we wish to consider.  What we show in the
rest of the paper is experimental results which show that even for
much smaller values of $n$ and $m$, a random sample produces graphs
that are difficult for isomorphism testers.  We combine random
generation of $\txor$ formulas with a filtering process which is aimed
at improving the likelihood of getting locally-consistent instances.
We describe this in the next section.

\section{Experimental Setup}\label{sec:setup}

Section~\ref{sec:construction} established a theoretical result that
shows that a random graph constructed in a particular way has the
properties of being asymmetric and having high Weisfeiler-Leman
dimension.  In outline, we want to start
with a random homogeneous $\txor$ formula on $n$ variables with
$\alpha n$ clauses, i.e.\ a random 3-uniform hypergraph on $n$
vertices and convert it into the graph $G_\phi$.  This graph is
asymmetric if $\phi$ is uniquely satisfiable (which occurs with high
probability).  Moreover $G_\phi$ has Weisfeiler-Leman dimension at
least $k$ if $\phi$ is $3k$-locally satisfiable, an event that also
occurs with high probability.  We now describe an experimental setup for producing such
graphs (with up to a few thousand nodes) by starting with a random
formula and applying a succession of filters.  In the process we apply a
number of heuristics in addition to the theoretical approach outlined
above.  To motivate and justify these heuristics, we now break up the
construction in a slightly different way.  

\paragraph{Asymmetry}
Consider $\phi$, a homogeneous $\txor$ formula with $n$ variables
$\mathcal{X} = \{X_1,\dots,X_n\}$ and $m$ clauses $\mathcal{C} = \{C_1,\ldots,C_m\}$.  We identify this
system with a bipartite graph $\Phi$ with vertices $\mathcal{C}$ on
the left and $\mathcal{X}$ on the right and an edge between $X$ and
$C$ if $X$ appears in the clause $C$.  Note that because $\phi$ is
homogeneous, the graph determines $\phi$ completely.  The construction
of the graph $G_{\phi}$ described above can now be broken up into two
steps.  In the first step, we produce a graph $G^1_{\phi}$ by
replacing each $X \in \mathcal{X}$ by two vertices $X^0$ and $X^1$ and
each $C \in \mathcal{C}$ with four vertices $\{C_{000},C_{011},C_{110},C_{101}\}$ and
connecting them as described above and illustrated in Figure~\ref{fig:cfi-gadget}.  In the second step, we augment
the graph $G^1_{\phi}$ with additional vertices $i_l, i_r, i_s$ for $i
\in \{1,\dots,n\}$. 

Note that the reason for the second step is effectively to impose a
linear order on the set of variables $\mathcal{X}$ and thereby ensure
that the only automorphisms of $G_{\phi}$ are the ones generated by
satisfying assignments to $\phi$.  Thus, if the graph $\Phi$ is itself
asymmetric, the second step is unnecessary as it is easily seen that
in this case the only automorphisms of $G^1_{\phi}$ are generated from
a satisfying truth assignment to $\phi$ by swapping $X^0$ and $X^1$
for all variables that are set to \texttt{true}.  How likely is it
that a random $\Phi$ (i.e.\ a random left-$3$-regular bipartite graph
with $m$ nodes on the left and $n$ on the right) is asymmetric?

We can think of $\Phi$ as a $3$-uniform hypergraph on $n$ nodes, with
$m$ edges.  It is easy to show (using the same methods that show that
a random graph is asymmetric (see~\cite[Chap.~9]{Bollobas})) that if
$m$ is roughly $n \log n$, then a random $3$-uniform hypergraph is,
indeed, asymmetric with probability going to $1$ as $n$ increases.
This is not the case when $m = \alpha n$ for constant $\alpha$.
However, our experiments show that in the range of values of $n$ we
worked with ($n$ up to about $3000$, and $m$ between $1$ and $5$), there was a
reasonably high probability of coming up with an asymmetric
hypergraph.  Moreover, if $\Phi$ is asymmetric, this is reasonably
quick to check with a tester such as \texttt{nauty/Traces} as it is
also highly probable that vertex refinement gives the orbit
partition.  It is only when $\Phi$ is converted to $G^1_{\phi}$ that
we get high Weisfeiler-Leman dimension.  Thus, for the purpose of the
experiments, instead of generating the graphs $G_{\phi}$ from $\Phi$,
we run a test on $\Phi$ to check if it has any non-trivial
automorphisms.  If it does, we discard it.  Otherwise, we construct
the graph $G^1_{\phi}$ and use that.

\paragraph{Unique Satisfiability}
Having generated a random homogeneous formula $\phi$, we wish to check
that it is uniquely satisfiable.  For this, we use a highly developed
SAT solver (\texttt{CryptoMiniSat 5}).  This SAT solver is
specifically optimized for cryptographic applications where the input
often contains clauses that are formed by taking the XOR, rather than
the disjunction, of a set of literals.  \texttt{CryptoMiniSat}
combines standard SAT solving methods (based on DPLL) with the
selective use of Gaussian elimination to attack such problems quickly.

In our filter, we express
$\phi$ as a conjunction of clauses where each clause is the XOR of
three variables.   We then test the satisfiability of $\phi'
\equiv \phi \land \bigvee_{X \in \mathcal{X}} X$.  That is, we add a
clause that is the disjunction of all variables in $\mathcal{X}$.  Of
course, $\phi'$ is satisfiable if, and only if, $\phi$ has a
satisfying assignment other than the all zeroes solution.  In other
words, $\phi'$ is satisfiable if, and only if, $\phi$ is \emph{not}
uniquely satisfiable.

\paragraph{Local Consistency}
We also want to ensure that the $\phi$ we select is $k$-locally
consistent for a sufficiently large value of $k$.  This is difficult
to check directly.  The problem of checking $k$-local consistency is
known to be hard, requiring time exponential in $k$ (see~\cite{Berkholz}) and we do not know of any good
implementations.  Instead, we used a simple heuristic that leverages
the specific capabilities of \texttt{CryptoMiniSat}.  Specifically,
this package allows us to turn the use of Gaussian elimination on and
off with an option.  We check the satisfiability of the formula
$\phi'$ by running \texttt{CryptoMiniSat} twice, once with Gaussian
elimination on and once with it off.  If the former is significantly
faster than the latter, we expect that the $\phi$ we have is a good
candidate.  Note, however, that this does not give us any bounds on
the value of $k$ for which $\phi$ might be $k$-locally consistent.

To justify this heuristic, note that the DPLL methods (with
clause-learning and restarts) as employed in modern SAT solvers are
subsumed by bounded-width resolution (see~\cite{AFT14}).  And formulas
that are highly locally consistent but not globally consistent are
exactly the ones that are difficult for bounded-width
resolution~\cite{Ats04}.  On the other hand, Gaussian elimination is a
method that specifically is fast for solving systems of linear
equations which may well be locally consistent~\cite{ABD09}.  Thus,
a formula on which Gaussian elimination is quick to determine
satisfiability but DPLL-based methods are slow is a prime candidate
for us.

\paragraph{Summary of Methodology}
In summary, our methodology for generating hard examples for
isomorphism testers is the following.  
\begin{enumerate}
\item For a fixed value of $n$ and
$m$ (roughly about $2n$), take a set $\mathcal{X}$ of $n$ variables.
\item Randomly select $m$ $3$-element subsets of $\mathcal{X}$ to form
  the left-$3$-regular bipartite graph $\Phi$.
\item Check (using \texttt{Traces}) to see if $\Phi$ has any
  non-trivial automorphisms.  If so, discard it.
\item If $\Phi$ has no non-trivial automorphisms, form the formula
  $\phi'$ by taking the conjunction of the clauses $\bigoplus_{X \sim
    C} X$ for each left-node $C$ of $\Phi$ along with the clause
  $\bigvee_{X \in \mathcal{X}} X$.  
\item Check if the formula $\phi'$ is satisfiable using the SAT solver
  \texttt{CryptoMiniSat} with Gaussian elimination option on.  If it
  is satisfiable, discard $\Phi$.
\item Run \texttt{CryptoMiniSat} on $\phi'$ a second time, with the
  Gaussian elimination option turned off.  If this takes significantly
  longer to determine $\phi'$ is unsatisfiable, then $\Phi$ is a prime
  candidate for the construction.
\item From $\Phi$, obtain the graph $G^1_{\phi}$ by replacing each
  node $C$ on the left-hand side with four nodes, and each node $X$ in
  the right-hand side with two nodes and connecting them as described
  earlier. 
\end{enumerate}

\section{Experimental Results}\label{sec:results}
We can report on three sets of experimental results, using the
construction described in the previous section.  It should be
noted that the main parameter that can be varied in the construction
is the ratio $m/n$ where $m$ is the number of clauses and $n$ the
number of variables in the $\txor$ formula.  The ratio needs to be at
least $1$ to guarantee that the constructed formula is uniquely
satisfiable.  The closer it is to $1$, the less likely it is to be
uniquely satisfiable.  Indeed, experimental runs show that at smaller
values we had to sample from the distribution $\mathcal{H}(m,n)$ many
times over to find a uniquely satisfiable instance.  On the other
hand, the larger the value of $m/n$, the harder it is to find locally
satisfiable instances.  While the theoretical results guarantee that
the probability of finding such instances increases with $n$, it
clearly does so more slowly for large values of $m/n$.  Hence, in
practice, one needs to fine tune the right value of the ratio to get
good results.

It should also be noted that our construction does not determine the
actual WL-dimension of the graphs.  This seems to be a much harder
computational problem than testing isomorphism itself
(see~\cite{Berkholz} for bounds on the related problem of determining
$k$-local consistency).  Thus, while the heuristic filters we use are
\emph{likely} to produce graphs of large WL-dimension, we are unable
to actually state the dimension of the graphs produced.

\paragraph{First Set.}
The procedure for constructing graphs described in the previous
section was run on a cloud server, with the specification given in
Table~\ref{table:test_env}, during April-June 2017.  The results show
that graphs of a few hundred nodes produced using this procedure are
very difficult for \texttt{Traces} in the sense that in most
cases, at this size, the system timed out (with a timeout set at 3
hours) and failed to identify the automorphism orbits.
\begin{table}[htbp!]
	\caption{Test Environment 1}
	\centering
	\label{table:test_env}
	\begin{tabular}{l l}
		\toprule
		Feature & Description  \\ 
		\midrule
		Host & DigitalOcean \\
		Operating System & Ubuntu 16.10 64-bit \\
		Memory & 2GB \\
		Disk & 20GB SSD \\
		CPU & 2 CPUs \\
		\bottomrule
	\end{tabular}
\end{table}

Some results of test runs on graphs produced by our construction are
shown in Figure~\ref{fig:px}.  These plot the time taken to run
\texttt{Traces} on graphs with $n$ nodes ($n$ being the
horizontal axis).  The plot on the left of the figure is for
graphs produced from $\txor$ formulas with $n$ variables and $m=n$
clauses.  The plot on the top right gives similar times for graphs
produced from  $\txor$ formulas with $m = 2n$.  Here, virtually all
graphs we were able to produce with over 5000 nodes timed out on
\texttt{Traces}.

\begin{figure}[htbp!] 
	\centering
	\begin{minipage}{.4\textwidth}
		\begin{tikzpicture}[trim left=0cm, scale=0.8]
		\begin{axis}[
		xlabel=nodes,
		ylabel=t sec,
		grid=both
		]
		\addplot[
		scatter,only marks,scatter src=explicit symbolic,
		scatter/classes={
			a={mark=o,blue}
		}
		]
		table[x=x,y=y,meta=label]{
			x    y    label
			36	0.000504970550537	a
			42	0.000595092773438	a
			48	0.000430107116699	a
			54	0.000416994094849	a
			60	0.000868082046509	a
			120	0.0010769367218	a
			180	0.00196194648743	a
			240	0.0024151802063	a
			300	0.00288486480713	a
			360	0.00310897827148	a
			420	0.00304412841797	a
			480	0.0377099514008	a
			540	0.822341918945	a
			600	1.99121904373	a
			606	0.287495136261	a
			612	0.0140619277954	a
			618	0.919209003448	a
			624	0.014004945755	a
			630	0.0803589820862	a
			636	12.7223041058	a
			642	12.3259220123	a
			648	249.624564171	a
			654	13.3357129097	a
			660	13.0116429329	a
			666	0.671919822693	a
			672	12.6001150608	a
			678	0.0159549713135	a
			684	1.46877002716	a
			708	286.804507017	a
			714	14.3620049953	a
			720	7.38974690437	a
			726	12.9122009277	a
			792	0.103566884995	a
			834	1.10935306549	a
		};
		\end{axis}
		\end{tikzpicture}
		
	\end{minipage}
	\begin{minipage}{.4\textwidth}
		\centering
		\begin{tikzpicture}[trim left=-1cm, scale=0.8]
		\begin{axis}[
		xlabel=nodes,
		ylabel=t sec,
		grid=both,
		extra y ticks={900},
		extra tick style={grid=major, grid style={dotted, cyan}},
		extra y tick labels={\hspace{5em}$timeout$},
		extra x tick labels={},
		extra y tick style={grid=none},
		ymax=900
		]
		\addplot[
		scatter,only marks,scatter src=explicit symbolic,
		scatter/classes={
			a={mark=o,blue}
		}
		]
		table[x=x,y=y,meta=label]{
			x    y    label
			60	0.00088095664978	a
			70	0.000917196273804	a
			80	0.000850915908813	a
			90	0.000887870788574	a
			100	0.000802040100098	a
			200	0.000893115997314	a
			300	0.00113701820374	a
			400	0.00194406509399	a
			500	0.00256896018982	a
			600	0.00155901908875	a
			700	0.00173497200012	a
			800	0.00396299362183	a
			900	0.00400900840759	a
			1000	0.00220799446106	a
			1100	0.0236101150513	a
			1200	0.00546789169312	a
			1300	0.00626587867737	a
			1400	0.00990104675293	a
			1500	0.0394530296326	a
			1600	0.0112538337708	a
			1700	0.0123739242554	a
			1800	0.0139899253845	a
			1900	0.0278260707855	a
			2000	0.0534319877625	a
			2100	0.00984311103821	a
			2200	0.0178880691528	a
			2300	1.90173387527	a
			2400	2.04471302032	a
			2500	0.200901985168	a
			2600	0.327310085297	a
			2700	0.923086881638	a
			2800	900	a
			2900	2.30818605423	a
			3000	2.48632717133	a
			3100	23.9029829502	a
			3200	22.727894783	a
			3300	23.0080840588	a
			3400	0.441438913345	a
			3500	0.49547791481	a
			3600	577.676494122	a
			3700	900	a
			3800	900	a
			3900	29.6461250782	a
			4000	570.171611071	a
			4100	609.037957907	a
			4200	609.872159004	a
			4300	3.48693990707	a
			4400	637.644945145	a
			4500	900	a
			4600	900	a
			4700	900	a
			4800	900	a
			4900	900	a
			5000	715.732777119	a
			6000	900	a
			7000	900	a
			8000	900	a
			9000	900	a
			10000	900	a
		};
		\end{axis}
		\end{tikzpicture}
	\end{minipage}
	\captionsetup{justification=centering}
	\caption{Left: n=m (con\_n).\\
		Right: n=2m (con\_2n). }
	\label{fig:px}
\end{figure}
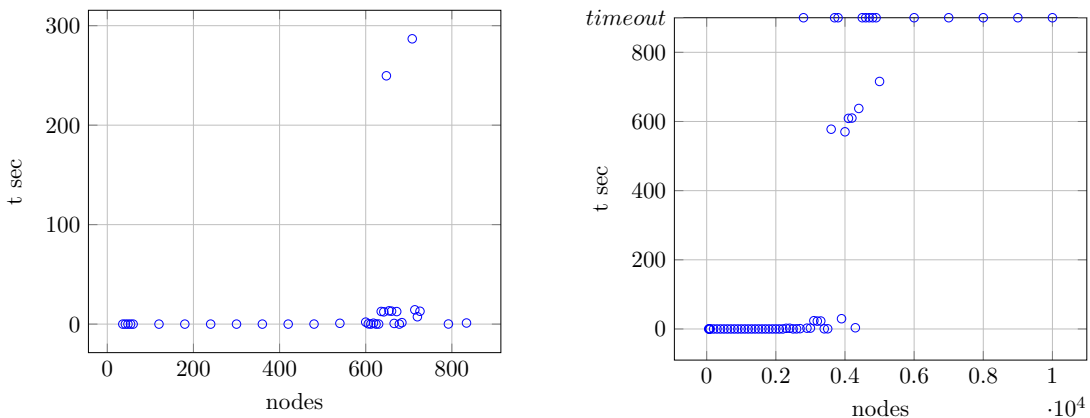

The complete data, including the graphs constructued, from this set of
experiments is available at
\url{https://github.com/kkcam/graph-ismorphism}.  Some explanation of
the nomenclature might be helpful.  The graphs are classified
according to the ratio $m/n$ used in their construction.  For
instance \texttt{con\_2n} is the collection of graphs with $m=2n$, and \texttt{con\_n}
is the collection of graphs with $m=n$.  There is also a package
\texttt{con\_sml} which contains for each $n$ the graph with the
smallest ratio $m/n$ for which we were able to obtain a uniquely
satisfiable $\txor$ formula, which also gives an asymmetric bipartite
graph $\Phi$.

\paragraph{Second Set}

 In February 2019, we ran a second set of
experiments.  These used the same database of graphs produced by the
construction for the first set.  This time the virtual machine setup
was as described in Table~\ref{table:test_env2}.  Again, with
\texttt{Traces}, most of the larger graphs timed out.  However, we
also ran the same set of graphs through \texttt{nauty}, \texttt{bliss}
and \texttt{conauto}, and all of these showed much better performance than
\texttt{Traces} on the large graphs in this collection.

\begin{table}[htbp!]
	\caption{Test Environment 2}
	\centering
	\label{table:test_env2}
	\begin{tabular}{l l}
		\toprule
		Feature & Description  \\ 
		\midrule
		Host & B2s Azure VM\\
		Operating System & Ubuntu 18.04 LTS \\
		Memory & 8GB \\
		Disk & 20GB SSD \\
		CPU & 2 CPUs \\
		\bottomrule
	\end{tabular}
\end{table}

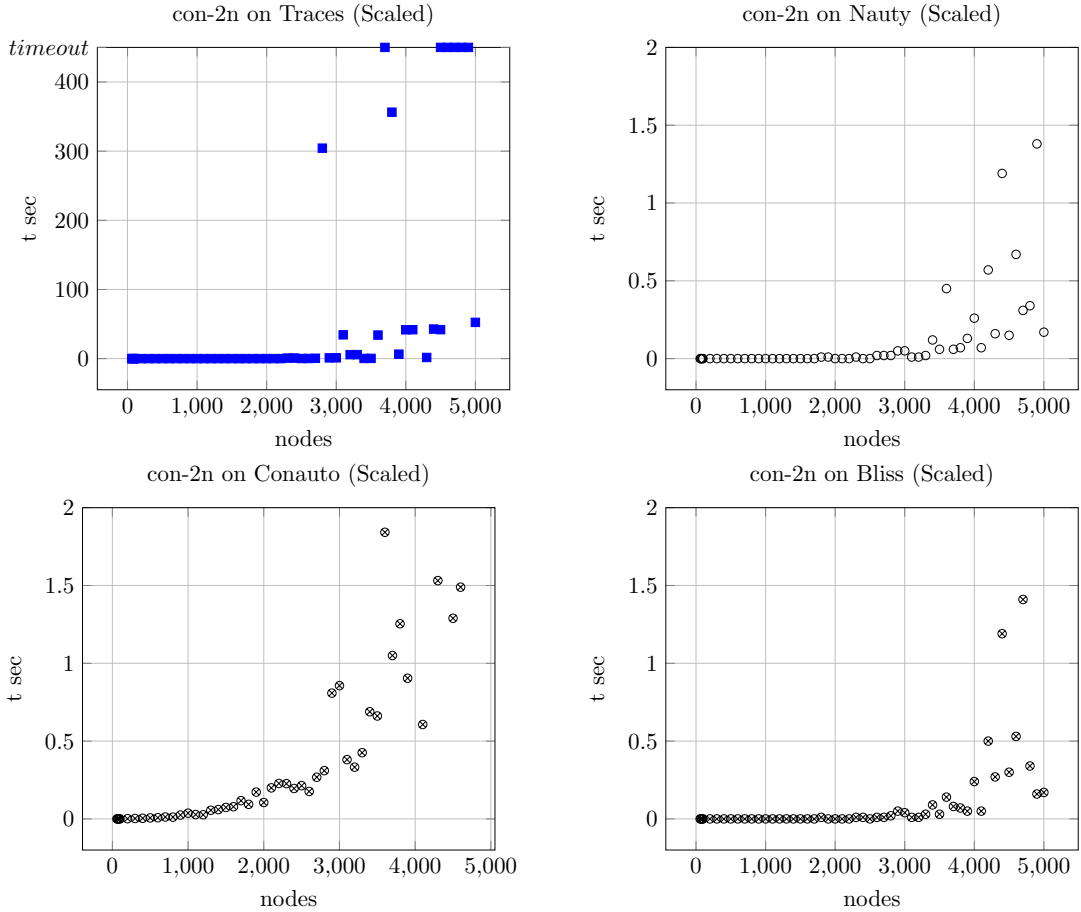
\begin{figure}[htb!] 
	\centering
        \begin{minipage}{.45\textwidth}
	\centering
	\begin{tikzpicture}[trim left=0cm, scale=0.8]
	\begin{axis}[
		xlabel=nodes,
		ylabel=t sec,
		legend pos=north west,
		grid=both,
		extra y ticks={450},
		extra tick style={
			grid=major,
			grid style={
				dotted,
				cyan
			}
		},
		extra y tick labels={$timeout$},
		extra x tick labels={},
		y label style={
			at={(axis description cs:0.05,.5)},
			anchor=south},
		extra y tick style={grid=none},
		ymax=450,
		title=con-2n on Traces (Scaled)
	]
	\addplot[
		scatter,
		only marks,
		scatter src=explicit symbolic,
		scatter/classes={
			traces={mark=square*,blue},
			conauto={mark=triangle*,red},
			nauty={mark=o,draw=black,fill=black},
			bliss={mark=otimes,draw=black,fill=black}
		}
	]
	table[x=x,y=y,meta=label]
	{
		x    y    label
60	0	traces
70	0	traces
80	0	traces
90	0	traces
100	0	traces
200	0	traces
300	0	traces
400	0	traces
500	0	traces
600	0	traces
700	0	traces
800	0	traces
900	0	traces
1000	0	traces
1100	0.01	traces
1200	0	traces
1300	0	traces
1400	0	traces
1500	0.02	traces
1600	0	traces
1700	0	traces
1800	0.01	traces
1900	0.01	traces
2000	0.04	traces
2100	0.01	traces
2200	0.01	traces
2300	0.78	traces
2400	0.83	traces
2500	0.14	traces
2600	0.21	traces
2700	0.51	traces
2800	304.28	traces
2900	1.08	traces
3000	1.18	traces
3100	34.41	traces
3200	5.74	traces
3300	5.76	traces
3400	0.29	traces
3500	0.34	traces
3600	34.05	traces
3700	450	traces
3800	356.32	traces
3900	6.42	traces
4000	41.68	traces
4100	41.82	traces
4500	41.89	traces
4300	1.67	traces
4400	42.73	traces
4500	450	traces
4600	450	traces
4700	450	traces
4800	450	traces
4900	450	traces
5000	52.42	traces
	};
	\end{axis}
	\end{tikzpicture}
        \end{minipage} 
        \begin{minipage}{.45\textwidth}
	\centering
	\begin{tikzpicture}[trim left=0cm, scale=0.8]
	\begin{axis}[
	xlabel=nodes,
	ylabel=t sec,
	legend pos=north west,
	grid=both,
	extra tick style={
		grid=major,
		grid style={
			dotted,
			cyan
		}
	},
	extra x tick labels={},
	y label style={
		at={(axis description cs:0.05,.5)},
		anchor=south},
	extra y tick style={grid=none},
	ymax=2,
	title=con-2n on Nauty (Scaled)
	]
	\addplot[
	scatter,
	only marks,
	scatter src=explicit symbolic,
	scatter/classes={
		traces={mark=square*,blue},
		conauto={mark=triangle*,red},
		nauty={mark=o,draw=black,fill=black},
		bliss={mark=otimes,draw=black,fill=black}
	}
	]
	table[x=x,y=y,meta=label]
	{
		x    y    label
60	0	nauty
70	0	nauty
80	0	nauty
90	0	nauty
100	0	nauty
200	0	nauty
300	0	nauty
400	0	nauty
500	0	nauty
600	0	nauty
700	0	nauty
800	0	nauty
900	0	nauty
1000	0	nauty
1100	0	nauty
1200	0	nauty
1300	0	nauty
1400	0	nauty
1500	0	nauty
1600	0	nauty
1700	0	nauty
1800	0.01	nauty
1900	0.01	nauty
2000	0	nauty
2100	0	nauty
2200	0	nauty
2300	0.01	nauty
2400	0	nauty
2500	0	nauty
2600	0.02	nauty
2700	0.02	nauty
2800	0.02	nauty
2900	0.05	nauty
3000	0.05	nauty
3100	0.01	nauty
3200	0.01	nauty
3300	0.02	nauty
3400	0.12	nauty
3500	0.06	nauty
3600	0.45	nauty
3700	0.06	nauty
3800	0.07	nauty
3900	0.13	nauty
4000	0.26	nauty
4100	0.07	nauty
4200	0.57	nauty
4300	0.16	nauty
4400	1.19	nauty
4500	0.15	nauty
4600	0.67	nauty
4700	0.31	nauty
4800	0.34	nauty
4900	1.38	nauty
5000	0.17	nauty
	};
	\end{axis}
	\end{tikzpicture}
        \end{minipage}

        \begin{minipage}{.45\textwidth}
	\centering
	\begin{tikzpicture}[trim left=0cm, scale=0.8]
	\begin{axis}[
	xlabel=nodes,
	ylabel=t sec,
	legend pos=north west,
	grid=both,
	extra tick style={
		grid=major,
		grid style={
			dotted,
			cyan
		}
	},
	extra x tick labels={},
	y label style={
		at={(axis description cs:0.05,.5)},
		anchor=south},
	extra y tick style={grid=none},
	ymax=2,
	title=con-2n on Conauto (Scaled)
	]
	\addplot[
	scatter,
	only marks,
	scatter src=explicit symbolic,
	scatter/classes={
		traces={mark=square*,blue},
		conauto={mark=triangle*,red},
		nauty={mark=o,draw=black,fill=black},
		conauto={mark=otimes,draw=black,fill=black}
	}
	]
	table[x=x,y=y,meta=label]
	{
		x    y    label
		60	0.000173	conauto
		70	0.000305	conauto
		80	0.000225	conauto
		90	0.000351	conauto
		100	0.000307	conauto
		200	0.00157	conauto
		300	0.002982	conauto
		400	0.004113	conauto
		500	0.006676	conauto
		600	0.007868	conauto
		700	0.012468	conauto
		800	0.011749	conauto
		900	0.025174	conauto
		1000	0.037079	conauto
		1100	0.028312	conauto
		1200	0.026808	conauto
		1300	0.055838	conauto
		1400	0.059993	conauto
		1500	0.073944	conauto
		1600	0.078582	conauto
		1700	0.117357	conauto
		1800	0.093859	conauto
		1900	0.172366	conauto
		2000	0.105528	conauto
		2100	0.199629	conauto
		2200	0.227976	conauto
		2300	0.226788	conauto
		2400	0.195812	conauto
		2500	0.213442	conauto
		2600	0.176564	conauto
		2700	0.267981	conauto
		2800	0.310169	conauto
		2900	0.808437	conauto
		3000	0.856482	conauto
		3100	0.381118	conauto
		3200	0.332706	conauto
		3300	0.425531	conauto
		3400	0.688712	conauto
		3500	0.661475	conauto
		3600	1.841973	conauto
		3700	1.05006	conauto
		3800	1.254185	conauto
		3900	0.904109	conauto
		4000	2.528401	conauto
		4100	0.606763	conauto
		4200	2.212026	conauto
		4300	1.53134	conauto
		4400	5.458036	conauto
		4500	1.28925	conauto
		4600	1.489192	conauto
		4700	68.823918	conauto
		4800	14.367406	conauto
		4900	4.261737	conauto
		5000	2.43648	conauto
	};
	\end{axis}
	\end{tikzpicture}
        \end{minipage}
        \begin{minipage}{.45\textwidth}
	\centering
	\begin{tikzpicture}[trim left=0cm, scale=0.8]
	\begin{axis}[
	xlabel=nodes,
	ylabel=t sec,
	legend pos=north west,
	grid=both,
	extra tick style={
		grid=major,
		grid style={
			dotted,
			cyan
		}
	},
	extra x tick labels={},
	y label style={
		at={(axis description cs:0.05,.5)},
		anchor=south},
	extra y tick style={grid=none},
	ymax=2,
	title=con-2n on Bliss (Scaled)
	]
	\addplot[
	scatter,
	only marks,
	scatter src=explicit symbolic,
	scatter/classes={
		traces={mark=square*,blue},
		conauto={mark=triangle*,red},
		nauty={mark=o,draw=black,fill=black},
		bliss={mark=otimes,draw=black,fill=black}
	}
	]
	table[x=x,y=y,meta=label]
	{
		x    y    label
		60	0	bliss
		70	0	bliss
		80	0	bliss
		90	0	bliss
		100	0	bliss
		200	0	bliss
		300	0	bliss
		400	0	bliss
		500	0	bliss
		600	0	bliss
		700	0	bliss
		800	0	bliss
		900	0	bliss
		1000	0	bliss
		1100	0	bliss
		1200	0	bliss
		1300	0	bliss
		1400	0	bliss
		1500	0	bliss
		1600	0	bliss
		1700	0	bliss
		1800	0.01	bliss
		1900	0	bliss
		2000	0	bliss
		2100	0	bliss
		2200	0	bliss
		2300	0.01	bliss
		2400	0.01	bliss
		2500	0	bliss
		2600	0.01	bliss
		2700	0.01	bliss
		2800	0.02	bliss
		2900	0.05	bliss
		3000	0.04	bliss
		3100	0.01	bliss
		3200	0.01	bliss
		3300	0.03	bliss
		3400	0.09	bliss
		3500	0.03	bliss
		3600	0.14	bliss
		3700	0.08	bliss
		3800	0.07	bliss
		3900	0.05	bliss
		4000	0.24	bliss
		4100	0.05	bliss
		4200	0.5	bliss
		4300	0.27	bliss
		4400	1.19	bliss
		4500	0.3	bliss
		4600	0.53	bliss
		4700	1.41	bliss
		4800	0.34	bliss
		4900	0.16	bliss
		5000	0.17	bliss
	};
	\end{axis}
	\end{tikzpicture}        
        \end{minipage}
	\captionsetup{justification=centering}
	\caption{Test run of first set of graphs on four different
          isomorphism solvers.}
	\label{fig:second}
\end{figure}

As a sample, we produce in Figure~\ref{fig:second} the timing results on the graphs in the
package \texttt{con\_2n} for each of the four isomorphism solvers.
The timeout is set at 120 m inutes and which can be seen
to occur frequently for \texttt{Traces}.  It should be noted that most
of the timeouts occurred due to memory limitations.  It seems
\texttt{Traces} requires large amounts of memory to process these
graphs and the swapping required is what leads to the process timing
out.  While the other solvers (in particular \texttt{bliss}) were able
to resolve the graphs quickly, they do show fast growth in
run times as the graphs get larger.  This is explored further in the
third set of experiments.

\begin{figure}[htb!]
  \begin{center}
    \includegraphics[width=0.8\linewidth]{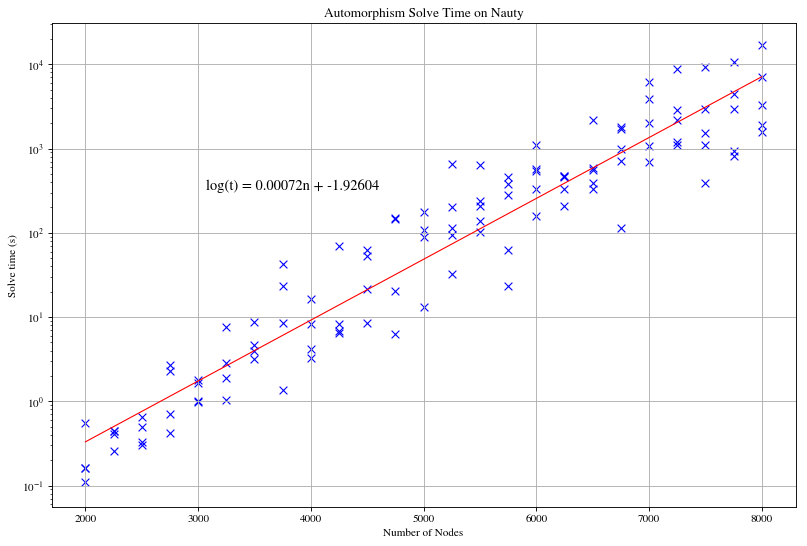}
  \end{center}

  \begin{center}
    \includegraphics[width=0.8\linewidth]{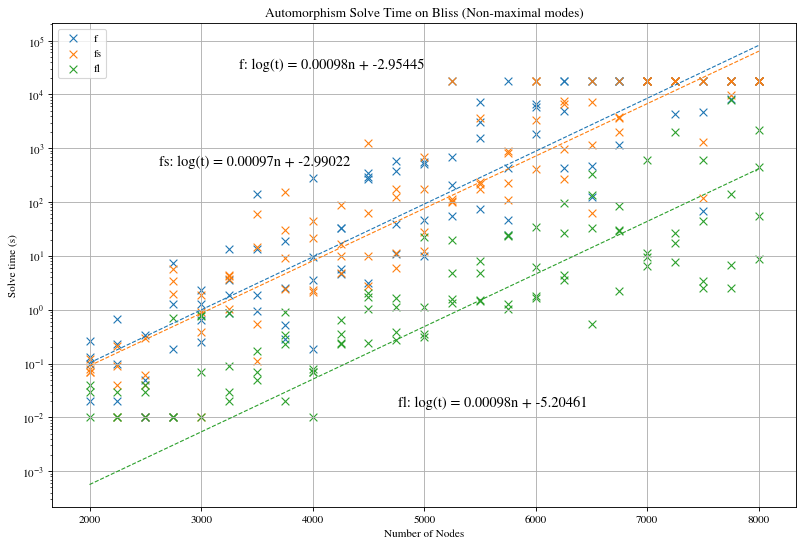}
  \end{center}
\caption{running times on \texttt{nauty} and \texttt{bliss}} \label{fig:yeung}
\end{figure}

\paragraph{Third Set}

We can report on a third set of experiments
performed by Richie Yeung in January-March 2019.  The full data on this can be found at
\url{https://github.com/y-richie-y/sat_cfi/}.  This generated a new
collection of graphs using the same protocol as described in
Section~\ref{sec:setup}.  These were run again through
\texttt{Traces}, \texttt{nauty}, \texttt{bliss} and \texttt{conauto}.
Graphs with up to 8000 nodes were generated with values of $m/n$ in
the range of $1.5$-$2$.  Once again, \texttt{Traces} frequently
(almost invariably) times out on the larger graphs.  The performance
of the other three solvers is better, but exhibits fast growth in
running time.  For example we exhibit the results for \texttt{nauty} and
\texttt{bliss} in Figure~\ref{fig:yeung}, with a
logarithmic scale on the $y$-axis for running time.  This is highly
indicative of exponential growth in running time.  Once again,
\texttt{bliss} proved to be the fastest of the solvers.  However, the
performance depends heavily on which target cell selection heuristic
is used.  As \texttt{bliss} allows the use of different heuristics by
setting parameters, results three different heuristics are displayed in
Figure~\ref{fig:yeung}, in different colours.  The best
performance is for \texttt{fl}, which is ``first largest non-singleton
cell''. 

\paragraph{Discussion}
There are some important points one should highlight from the
experimental results.  The first is that \texttt{Traces} performs
significantly worse on the graphs we construct than any of the other
solvers.  One possible explanation for this is the fact that the
fundamental algorithm in \texttt{Traces} is a breadth-first search
procedure of the individualization tree.  Such a procedure may require
shallowe trees but may, in principle, be more memory-intensive than a depth-first search.  An
important way that \texttt{Traces} avoids this drawback is the early
identification of symmetries in the graph and using this to prune the
search space.  It is possible that the lack of symmetries in our
graphs makes this pruning impossible leading to the solver running out
of memory and timing out as a result.  The construction described by
Neuen and Schweitzer~\cite{NeuenS-ESA} is also aimed at constructing
graphs which are asymmetric and have high WL-dimension.  They also,
similarly, report that \texttt{Traces} is rather slower on their
graphs than other solvers.  In contrast, Yeung reports that his
implementation of the Neuen-Schweitzer construction yields graphs
on which \texttt{Traces} performs faster than \texttt{nauty}.  This
warrants further investigation.

Apart from \texttt{Traces}, an important distinction between the other
solvers tested is their cell selection strategy.  One of them,
\texttt{bliss}, explicitly allows the user to choose the strategy in a
call to the system.  As we have seen, the choice of strategy can have
a significant effect on the performance of the solvers.  The results
of the third set of experiments, especially on \texttt{bliss}, demonstrate the effect starkly.  It
would be instructive to understand how these cell selection strategies
interact with the construction we have presented.

Our theoretical construction shows that there exists a family of
graphs on which any solver based on individualization and refinement,
along with factoring by symmetries, will take exponential time, no
matter what cell selection strategy is used.  Furthermore, it shows
that sampling graphs at random from the distribution we describe will
produce such graphs with high probability.  We cannot verify that the
graphs we select do indeed have high WL-dimension, which is why we
need experimental validation, and the results do strongly suggest that
the growth rate, on any solver, is exponential.  In the first set of
experiments, we constructed graphs with parameter $m/n \leq 2$ only up
to about 5000 nodes.  For larger graphs, larger ratios were used.  The
third set of experimental results extended the construction of graphs
with small ratio up to about 8000 nodes (e.g.\ $n=800$, $m=1600$),
and the increase in running time is striking.  The main reason for
using larger ratios to generate the larger graphs in the first case
was that at small ratios, finding large $\txor$ formulas that are
uniquely satisfiable becomes difficult, requiring large numbers of
re-trials with fresh sampling.  When this is combining with two runs
of a SAT solver for each sampled formula, the time becomes
prohibitive.  Also, as we are not using the asymmetry gadgets
described in Section~\ref{sec:construction}, we are relying on
checking that the random 3-left-regular hypergraph we select is itself
asymmetric, and this may also involve repeated trials.  Here the
probability of success decreases with $n$ for a constant ratio.
The protocol was improved in the third set by checking
unique satisfiability by a direct rank computation.  Then, the SAT
solver check was only performed for those formulas already known to be
uniquely satisfiable, merely to record the difference caused by the
use of Gaussian elimination.

\section{Conclusion}
We have described a theoretical construction of graphs that are provably
difficult for a isomorphism solvers such as \texttt{nauty} and
\texttt{Traces}.  We have examined the construction experimentally and
the results indicate that the graphs produced do indeed show
exponential growth rates in running time on these solvers.

The main theoretical result combines known lower bounds for local
consistency of $\txor$ formulas with a construction inspired by the
graphs of Cai-F{\"u}rer-Immerman and the related multipede
construction to give a family of graphs which are provably asymmetric
and of linear Weisfeiler-Leman dimension.  This ensures that the
running time grows exponentially with the size of the graphs.  Our
result also shows that a randomly constructed $\txor$ instance is
likely to yield a difficult graph with high probability.  That is, the
probability tends towards $1$ as the graphs get larger.

The experimental setup uses SAT solver technology to create a series
of filters which, combined with the random generation of $\txor$ formulas
produces graphs to follow the theoretical procedure.  For the
experimental set-up, we dropped some of the theoretical guarantees on
asymmetry and local consistency and replaced them with heuristic
tests.  This is because we are unable to verify directly the
WL-dimension of the graphs constructed.  

The results show that our method quickly and consistently produces
graphs that are difficult for \texttt{Traces}.  Experiments with other
solvers also support the conclusion that the growth rate of the time
taken is exponential.  This is comparable with the construction of
hard graphs in~\cite{NeuenS-ESA}.

\bibliographystyle{plainurl}
\bibliography{references}

\end{document}